\documentclass{IEEEtran}
\usepackage[breaklinks]{hyperref}
\usepackage{color, dsfont, bm, epsfig, graphicx, cite, multirow, multicol, comment, mathtools}
\usepackage[ruled, noend]{algorithm2e}
\usepackage{amsmath,amsfonts,amsthm,amssymb,subfigure}
\usepackage{stfloats,booktabs,threeparttable,tabularx,placeins}
\usepackage{mathrsfs}

\newcommand{\st}[1]{\text{s.t.}}
\renewcommand{\arraystretch}{1.2}

\newtheorem{definition}{Definition} 
\newtheorem{lemma}{Lemma}

\newtheorem{remark}{Remark}
\usepackage{xcolor}
\usepackage{transparent} 
\usepackage{textcomp}
\usepackage{newtxtext}

\begin{document}

\title{Rethinking Channel Estimation for Fluid Antenna Systems: A Physics-Informed Operator Learning}

\author{
    Zixin Wang,~\IEEEmembership{Member,~IEEE},
    Lizhong Zheng,~\IEEEmembership{Fellow,~IEEE},
    Yuanming Shi,~\IEEEmembership{Senior Member,~IEEE},
    Kaibin Huang,~\IEEEmembership{Fellow,~IEEE},
    and Khaled B. Letaief,~\IEEEmembership{Fellow,~IEEE}
    \thanks{Z. Wang, and K. B. Letaief are with the Department of Electronic and Computer Engineering, The Hong Kong University of Science and Technology, Hong Kong (e-mail: \{eewangzx,eekhaled\}@ust.hk). 
    L. Zheng is with the Department of Electrical Engineering and Computer Science, Massachusetts Institute of Technology, Cambridge, MA 02139, USA (e-mail: lizhong@mit.edu).
    Y. Shi is with the School of Information Science and Technology, ShanghaiTech University, Shanghai, 201210, China (E-mail: shiym@shanghaitech.edu.cn).
    K. Huang is with the Department of Electrical and Computer Engineering, The University of Hong Kong, Hong Kong (e-mail: huangkb@hku.hk).
}
}

\maketitle
\begin{abstract}
Fluid antenna systems (FAS) leverage continuous reconfigurable apertures to transcend the spatial limits of fixed-position arrays, yet acquiring full-aperture channel state information (CSI) from sparse observations remains a formidable bottleneck. 
\textcolor{black}{In this paper, we establish a continuous-field formulation for FAS by recasting channel estimation as a stochastic inverse field-reconstruction problem constrained by the projected Helmholtz model.} Through electromagnetic information theory (EIT)-based spectral leakage analysis, we demonstrate that finite-aperture reconstruction induces an intrinsic convolutional relationship in the spectral domain, rendering the estimation task a fundamental frequency-domain inverse deconvolution problem. 
To solve this, we propose a U-shaped Fourier neural operator (UFNO) specifically tailored for noisy, underdetermined wireless environments. 
\textcolor{black}{Furthermore, we leverage the Landau--Pollak theorem to derive an aperture-dependent EDoF bound that guides the spatial mode budget, and employ a physics-aware Sobolev loss to promote consistency in both channel values and spatial variations.} 
\textcolor{black}{Unlike conventional learning-based approaches, a single frozen UFNO can be directly evaluated at unseen spatial--frequency resolutions without resolution-specific retraining or explicit output interpolation.} Extensive simulations show that our framework significantly outperforms state-of-the-art benchmarks in accuracy while achieving \textcolor{black}{orders-of-magnitude lower inference latency than iterative solvers at high spatial--frequency resolutions}.
\end{abstract}


\section{Introduction}

\IEEEPARstart{M}{ultiple}-input multiple-output (MIMO) technology is fundamental to contemporary wireless networks, leveraging spatial degrees of freedom (DoFs) to enhance spectral efficiency and link reliability \cite{11433651, 10812728, pizzo2022fourier, 10925441}. 
However, conventional MIMO architectures are inherently limited by the Nyquist-based half-wavelength spacing rule, which restricts spatial diversity gains within a fixed physical aperture regardless of antenna density. 
To surpass these constraints, continuous-aperture MIMO (CAP-MIMO) has emerged as a transformative paradigm \cite{zhu2025electromagnetic, pizzo2022fourier}. By integrating metasurfaces or dense sub-wavelength arrays, CAP-MIMO enables ultra-fine electromagnetic (EM) control, effectively realizing a holographic radiating aperture \cite{zhu2025electromagnetic}. 
Nevertheless, translating the theoretical potential of CAP-MIMO into practice is hindered by several interrelated challenges. The extreme proximity of sub-wavelength elements induces severe mutual coupling and complex near-field interactions that complicate signal processing. 
Furthermore, acquiring full-space EM channel state information (CSI) necessitates high-dimensional channel matrices, where the resulting ``curse of dimensionality'' imposes prohibitive pilot and computational overheads. 
These difficulties are compounded by the excessive hardware costs and power consumption associated with scaling active elements, presenting substantial barriers to the large-scale commercialization of dense-array-centric architectures.

Against this backdrop, fluid antenna systems (FAS) introduce a transformative paradigm that sidesteps the limitations of static hardware densification \cite{9264694, 10753482, liao2025joint, 11193779}.
By enabling a limited number of movable antennas to dynamically traverse a continuous spatial manifold within a predefined aperture, FAS not only minimizes hardware overhead but also circumvents the analytical complexities and hardware redundancy inherent in metasurface-based CAP-MIMO.
The seminal work in \cite{10794752} first introduced the FAS concept, and demonstrated that even a small linear aperture (e.g., 0.5$\lambda$) with sub-wavelength port spacing can yield significant spatial diversity gains, even surpassing multi-antenna MRC systems.
Furthermore, the authors in \cite{9131873} revealed that FAS can achieve substantial ergodic capacity in both point-to-point and multi-user multiple-access scenarios.
However, FAS performance hinges on high-resolution CSI across the full spatial domain, demanding dense sub-wavelength measurements. Hardware constraints permit only a sparse subset, rendering channel estimation a severely underdetermined inverse problem.

Considering its sparse nature, a conventional approach is adopting compressed sensing (CS) technology\cite{10236898,10755170,10130117}.
By sending pilots from the random antenna ports, the channel estimation problem can be formulated as a sparse signal recovery problem, which can be solved by algorithms such as orthogonal matching pursuit (OMP) or sparse Bayesian learning (SBL). 
For instance, the authors in \cite{10236898} proposed a successive transmitter-receiver compressed sensing (STRCS) method by exploiting the angular-domain sparsity of multi-path components to reconstruct the continuous field-response information from finite measurements.
The authors in \cite{10755170} developed an SBL-based channel estimator by reformulating the estimation task into a basis selection problem, thereby achieving superior reconstruction accuracy with limited radio-frequency chains.
However, these methods are fundamentally limited by the discretization of the channel into a finite grid, leading to grid mismatch artifacts and high computational complexity in high-dimensional scenarios.
Besides, these methods highly rely on the sparsity assumption based on the shallow scattering environment, leading to limited implementation cases.

Conversely, learning-based approaches leverage data-driven architectures to characterize the nonlinear mapping from sparse pilot observations to full-space channel realizations, thereby shifting the burden of computationally intensive iterations to low-latency feed-forward inference. 
Beyond conventional neural network-driven channel extrapolation, which seeks to learn direct mappings from sparse inputs to high-dimensional CSI, the reconstruction of the channel field via generative modeling has emerged as a pivotal research frontier. 
Frameworks such as diffusion models \cite{tang2025accurate}, conditional generative adversarial networks (cGANs)\cite{10663271}, and variational autoencoders (VAEs)\cite{10495003} reformulate channel estimation as an inverse problem regularized by learned implicit priors. These methods have demonstrated a superior capacity for capturing intricate electromagnetic (EM) characteristics and spatial correlations at sub-wavelength scales.
Despite their reconstructive fidelity, diffusion-based methods necessitate multi-step iterative sampling, which incurs substantial computational latency. 
This delay poses a significant challenge in satisfying the stringent constraints of the channel coherence interval, particularly in high-mobility scenarios. 
Furthermore, while techniques based on 3D Gaussian Splatting (3DGS), such as RF-PGS\cite{zhang2025rf}, RF-3DGS\cite{zhang2024rf3dgs}, facilitate rapid rendering and geometric alignment, they are fundamentally limited to site-specific optimization. These approaches lack the generalization capability required for inference in unseen environments and often neglect the underlying physical operator constraints governing electromagnetic wave propagation.

Grounded in electromagnetic information theory (EIT) \cite{zhu2025electromagnetic, li2025spatio, 10751774, castellanos2024electromagnetic}, wireless channels are fundamentally modeled as continuous physical fields governed by Maxwell's equations and the associated Helmholtz equation, rather than as collections of isolated random variables. 
Although these fields exhibit an inherent low-rank structure in the wavenumber domain, existing solution methodologies encounter significant technical barriers. 
Specifically, numerical solvers, e.g., ray tracing, the finite element method, and the finite-difference method, are often computationally prohibitive. 
Moreover, they are constrained by a prohibitive trade-off between computational accuracy and efficiency, making them ill-suited for the millisecond-level real-time reconstruction required by 6G networks.
Alternatively, statistical kernel-based approaches\cite{li2025spatio}, including Gaussian process regression and kernel learning, have introduced physically inspired kernels, such as the electromagnetic correlation function. 
However, these methods essentially rely on fixed-kernel structures, which lead to high computational complexity for matrix inversion and struggle to adapt to complex, dynamic EM environments.
\textcolor{black}{Furthermore, most existing techniques lack resolution invariance, limiting direct channel reconstruction at spatial--frequency discretizations unseen during training.}

Inspired by the paradigm shift toward neural operator learning for partial differential equations (PDEs) \cite{li2020fourier}, we propose a physics-informed learning framework tailored for FAS channel estimation. \textcolor{black}{Unlike deterministic and noise-free PDE benchmarks where Fourier Neural Operators (FNOs) excel, FAS estimation presents an ill-posed inverse problem with an unknown stochastic excitation and severely undersampled observations.}
Standard FNOs can not be directly implemented for this problem due to their rigid spectral truncation, which lacks adaptability to varying noise levels, and a pervasive low-pass bias that suppresses critical sub-wavelength spatial features. 
{\color{black}To address these challenges, we reformulate FAS channel estimation as a continuous operator learning task and introduce a U-shaped FNO (UFNO) whose non-local architecture is motivated by the projected Helmholtz model. The primary contributions of this work are summarized as follows.
\begin{enumerate}
\item [$\bullet$]We recast FAS channel estimation as a stochastic inverse field-reconstruction problem constrained by the projected Helmholtz model. 
By shifting from discrete vector recovery to continuous field reconstruction, we utilize EIT to reveal that finite-aperture sensing induces an intrinsic spectral convolution. This analysis formalizes the channel estimation task as an inverse deconvolution problem within the spatio-frequency domain.
\item [$\bullet$]Beyond empirical modeling, we derive a closed-form upper bound on the effective degrees of freedom (EDoF) of the FAS channel via the Landau--Pollak theorem. This result characterizes the aperture-dependent spatial dimensionality and provides a physics-based reference for selecting the UFNO spatial mode budget. We further distinguish this ideal dimension count from SNR-dependent stable recovery through a finite-aperture singular-value analysis and multi-aperture simulations.
\item [$\bullet$]We develop a UFNO based channel estimation framework to overcome the limitations of conventional operators through two key mechanisms. First, a learnable soft-thresholding operator is introduced to dynamically suppress noise while preserving physical signal modes across varying signal-to-noise ratios (SNRs). 
Second, a dual-branch spectral-spatial structure is employed to mitigate low-pass bias, enabling the simultaneous capture of macroscopic low-frequency structures and high-frequency sub-wavelength details.
\item [$\bullet$]We propose a physics-aware Sobolev loss that penalizes discrepancies in both channel values and spatial--frequency variations. 
This regularization enhances the recovery of high-gradient scattering structures. Numerical results demonstrate that our framework achieves grid-agnostic reconstruction with orders-of-magnitude lower latency than model-based baselines.
\end{enumerate}
The rest of the paper is organized as follows. 
Section II introduces the system model and FAS channel modeling. 
Section III details the proposed UFNO framework, including its architecture and loss function design. Section IV presents simulation results comparing UFNO with benchmark methods. Finally, Section V concludes the paper.}

\section{System Model and Modeling of FAS Channels}
\textcolor{black}{This section establishes a physics-informed model in which the continuous aperture field is characterized by a projected stochastic Helmholtz relation.}
By analyzing the resulting spatio-frequency structure, we establish the motivation for operator learning.

\subsection{Channel Propagation as a Physical Field}
\begin{figure*}[t]
\centering
\includegraphics[width = \linewidth]{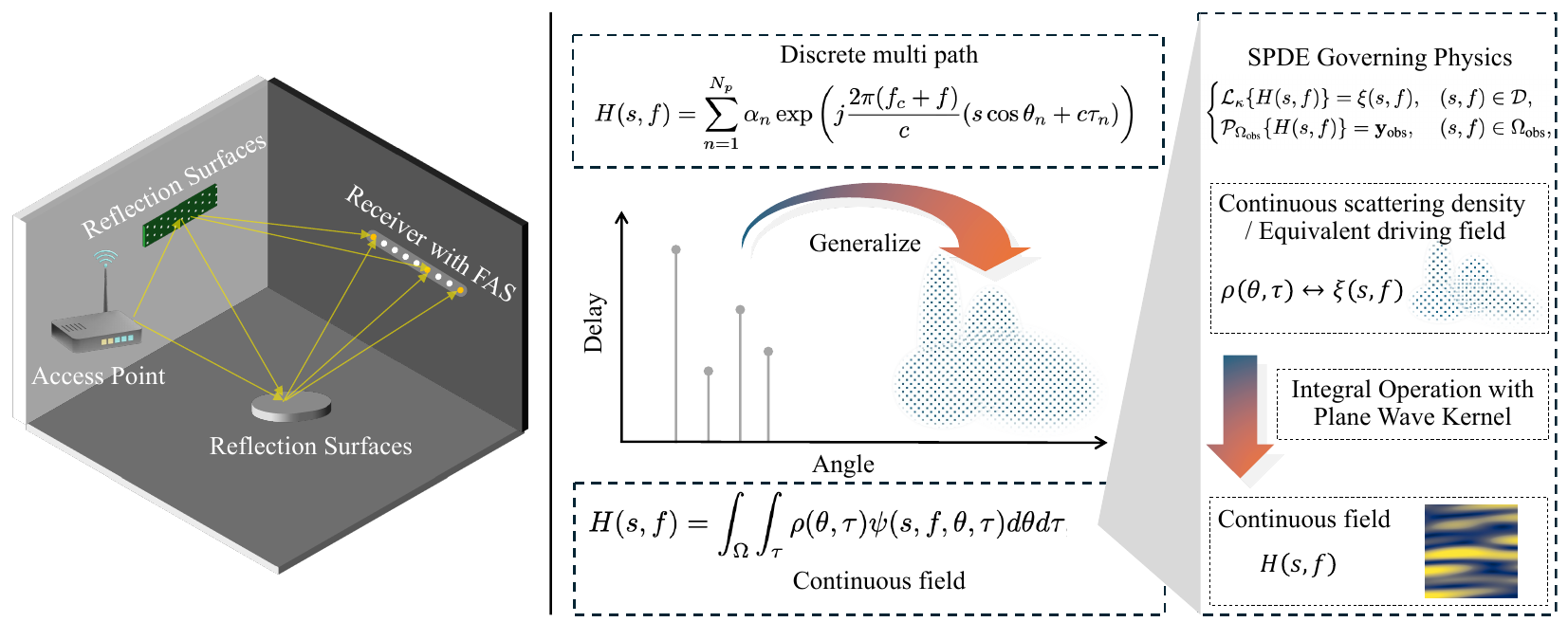}
\caption{The continuous-field paradigm for FAS channel modeling.
The left panel depicts the physical multipath environment with a fluid antenna receiver sampling the spatio-frequency manifold. 
The middle panel illustrates the generalization of the channel model from a discrete multipath summation to a continuous-field integral, representing the scattering density in the delay-angle domain. 
The right panel highlights the governing physics, where the channel field $H(s, f)$ is characterized as the solution to a SPDE driven by the equivalent scattering source $\xi(s, f)$.}
\label{fig::system}
\end{figure*}
Following the FAS framework in \cite{dong2025frequency}, we consider a wideband FAS architecture where a single-antenna transmitter serves a receiver (Rx) equipped with a one-dimensional continuous linear aperture $\mathcal{S} \subset \mathbb{R}$ of length $L$, as shown in Fig. \ref{fig::system}. 
The system operates over an orthogonal frequency division multiplexing (OFDM) framework across a bandwidth $B$, partitioned into $N_f$ subcarriers. The frequency domain is denoted as $\mathcal{B} \triangleq [-B/2, B/2]$. 
The spatial aperture is normalized as $L = W\lambda$, where $\lambda$ denotes the carrier wavelength and $W$ is the aperture dimension. In view of hardware efficiency, the Rx is constrained by a limited number of $M$ radio-frequency (RF) chains, which implies that only a sparse subset of the spatio-frequency manifold can be observed within any given coherence interval\footnote{In practice, block accumulation across multiple coherence intervals can be employed to increase the effective number of observations beyond $M$. Nevertheless, limited by available spectral resources, the effective observations remain sparse compared to $N_s \times N_f$.}.

Recall that the wireless channels exhibit inherent variability across both spatial and frequency domains due to multipath propagation, characterized by angular and delay spreads, a phenomenon termed spatio-frequency selectivity. 
Consequently, antenna displacement within the aperture $\mathcal{S}$ induces significant fluctuations in channel realizations. 
In this paper, we utilize the spatio-frequency characteristic (SFC) to jointly represent this diversity. 
The SFC models all potential propagation modes between regions over a specified bandwidth, thereby establishing a mapping from the input spatio-frequency-dependent pilot signals to the resulting channel responses.
\begin{definition}[SFC Field]
The SFC of the FAS-aided link are defined as a continuous complex-valued field $H(s, f)$ over the joint spatio-frequency domain $\mathcal{D} \triangleq [0, L] \times [-B/2, B/2]$. 
This field characterizes the input-output relationship of the propagation environment, such that for a source signal $x(f)$, the received signal $y(s, f)$ is given by
\begin{equation}
y(s, f) = \int_{\mathcal{B}} H(s, f') x(f') \delta(f - f') df' = H(s, f)x(f).
\end{equation}
\end{definition}
Conventionally, we can assume that there are $N_p$ discrete transmission paths from the fixed Tx to the Rx, then the aggregate SFC can be expressed as the superposition of plane waves:
\begin{equation*}
H(s, f) = \sum_{n=1}^{N_p} \alpha_n \exp\!\left( { j\, 2\pi(f_c+f)\!\left(\frac{s\cos\theta_n}{c}+\tau_n\right)} \right),
\end{equation*}
where $f_c$ and $c$ represent the carrier frequency and the velocity of light, respectively. 
The triplet $\{\alpha_n, \theta_n, \tau_n\}$ denotes the complex gain, angle of arrival (AoA), and propagation delay associated with the $n$-th path. 
In particular, the term $2\pi(f_c+f)\!\left(\dfrac{s\cos\theta_n}{c}+\tau_n\right)$ encapsulates the spatio-frequency dispersion. By expressing the spatial displacement as an equivalent propagation delay $s\cos\theta_n/c$ and the propagation delay as $\tau_n$, both terms share the same temporal unit. 
This pairing of the instantaneous angular frequency $2\pi(f_c+f)$ with the aggregate delay $(s\cos\theta_n/c + \tau_n)$ gives rise to a non-separable harmonic structure in $H(s,f)$. 
Consequently, the spectral selectivity induced by the delay spread $\tau_n$ is intrinsically coupled with the spatial variations across the aperture $L$. 
It is worth noting that for simplicity of analysis, we ignore the frequency-dependent waveform broadening and distortion\cite{4802185, 11185051} of each path component under the ultra-wideband scenarios as in \cite{11185051}.  

Despite the prevalence of discrete path models, their reliance on the isolated point-scatterer assumption is insufficient for characterizing near-field interactions or continuously distributed scattering surfaces, such as walls and large-scale structures. 
By approximating multipath components as idealized Dirac impulses in both the angular and delay domains, these models overlook energy diffusion and diffuse scattering arising from surface roughness. 
For FAS, which utilize sub-wavelength spatial positioning, this diffusion is a critical determinant of the physical width and resolvability of spatial gain profiles. 
Furthermore, the practical difficulty in accurately estimating the parameter triplet $\{\alpha_n, \theta_n, \tau_n\}$ for each path renders the high-fidelity reconstruction of $H(s,f)$ from sparse RF observations a fundamentally ill-posed inverse problem.

\subsection{EIT-based SPDE Formulation}
Building on the continuous-field perspective, we now establish the governing physics of the FAS channel by leveraging EIT, which treats the SFC $H(s,f)$ is treated not merely as a heuristic stochastic process, but as a physical field strictly governed by the laws of electrodynamics.
Specifically, the SFC is characterized as a continuous integral of the scattering density function, representing the aggregate EM response across the spatial and frequency manifolds, expressed as:
\begin{align}
H(s, f) = \int_{\Omega} \int_{\tau} \rho(\theta, \tau) \psi(s, f, \theta, \tau) d\theta d\tau,
\end{align}
where $\rho(\theta, \tau)$ denotes the scattering density function, representing the joint distribution of electromagnetic intensity across the angular domain $\Omega$ and delay domain $\tau$.
The kernel $\psi(s, f, \theta, \tau) = \exp\!\left(j\,2\pi(f_c+f)\!\left(\frac{s\cos\theta}{c}+\tau\right)\right)$ represents the phase kernel of a plane wave propagating with angle $\theta$ and delay $\tau$, which fully characterizes the plane-wave eigenfunctions of the Helmholtz equation.
Consequently, conventional discrete path models can be interpreted as sampled realizations of this continuous integral under finite DoF. 

Within a source-free 2D local propagation plane spanned by the array axis $s$ and its transverse axis $z$, 
the physical electromagnetic field rigorously satisfies the Helmholtz equation \cite{pizzo2022spatial,zhu2025electromagnetic}, i.e.,
\begin{equation}
    (\nabla_s^2 + \nabla_z^2 + \kappa^2(f))H(s, z, f) = 0,
\end{equation}
where $\kappa(f)=2\pi(f_c+f)/c$ is the wavenumber.
By evaluating this field on the 1D FAS continuous aperture at $z=0$, the spatial evolution of the aperture field is governed by the projected equation:
\begin{equation} \label{eq:helmholtz}
(\nabla_s^2 + \kappa^2(f))H(s, f) = -\left.\frac{\partial^2 H(s, z, f)}{\partial z^2}\right|_{z=0},
\end{equation}
where $\nabla_s^2$ denotes the 1D Laplacian operator (i.e., $\partial^2/\partial s^2$).
\textcolor{black}{Accordingly, we characterize channel estimation as a stochastic inverse field-reconstruction problem.}
Let $\mathcal{D} = \mathcal{S} \times \mathcal{B}$ denote the joint spatio-frequency domain. \textcolor{black}{The aperture field and its sparse observations satisfy the following coupled physical and measurement relations:}
\begin{equation} \label{eq:spde_system}
\begin{cases}
    \mathcal{L}_{\kappa} \{ H(s, f) \} = \xi(s, f), & (s, f) \in \mathcal{D}, \\
    \mathcal{P}_{\Omega_{\text{obs}}} \{ H(s, f) \} = \mathbf{y}_{\text{obs}}, & (s, f) \in \Omega_{\text{obs}},
\end{cases}
\end{equation}
where $\mathcal{L}_{\kappa} = \nabla_s^2 + \kappa^2(f)$ is the 1D Helmholtz differential operator, and
$\xi(s, f) \triangleq -\left.\frac{\partial^2 H(s, z, f)}{\partial z^2}\right|_{z=0}$
captures the transverse wave dynamics injected from the multipath environment.
$\mathcal{P}_{\Omega_{\text{obs}}}$ is the projection operator mapping the continuous field to the sparse set of discrete RF observations $\mathbf{y}_{\text{obs}}$ at coordinates $\Omega_{\text{obs}}$.
Since multipath components arrive from random angles with random complex gains, the transverse derivative term naturally behaves as a complex-valued stochastic process.
Accordingly, we model $\xi(s, f)$ as a stationary complex Gaussian random field with zero mean and covariance
\begin{equation}
\mathbb{E}[\xi(s, f) \xi^*(s', f')] = \sigma^2 R_{\xi}(s - s', f - f'),
\end{equation}
where $R_{\xi}(\cdot,\cdot)$ is induced by the angular-delay power spectrum of the scattering environment.

\subsection{From Physical Principles to Operator Learning}
To characterize the intrinsic structure of the reconstruction problem posed in \eqref{eq:spde_system}, we analyze it from two complementary perspectives: the spectral domain and the spatial domain.
Regarding the spectral domain, under the source-free condition (i.e., the homogeneous counterpart of \eqref{eq:helmholtz}), the field exhibits a strictly band-limited structure in classical wave theory \cite{10751774}, with spectral energy mathematically confined to a bounded support.
Theoretically, this property guarantees perfect field reconstruction via Nyquist-compliant sampling.
However, while the global EM field is band-limited, an FAS receiver only observes this field over a finite, continuous aperture $\mathcal{S} = [0, L]$.
Mathematically, this finite-aperture observation process is formulated as:
\begin{equation} \label{eq:spatial_window}
H_{\text{obs}}(s, f) = H_{\infty}(s, f) \cdot \text{rect}\left( \frac{s - L/2}{L} \right),
\end{equation}
where $H_{\infty}(s, f)$ denotes the ideal field of infinite extent, and $\text{rect}(\cdot)$ is the rectangular window function. 
In the wavenumber domain, this spatial truncation translates into a convolution:
\begin{equation} \label{eq:spectral_leakage}
\hat{H}_{\text{obs}}(k_s, f) = \hat{H}_{\infty}(k_s, f) \circledast \underbrace{L \text{sinc}(k_s L) e^{-j \pi k_s L}}_{\mathcal{W}(k_s)},
\end{equation}
where $\hat{H}_{\infty}(k_s, f)$ is the ideal band-limited spectrum supported on $[-\kappa, \kappa]$, $\mathcal{W}(k_s)$ is the Fourier transform of the spatial window, and $\circledast$ denotes convolution in the wavenumber domain.
Governed by the space-wavenumber uncertainty principle \cite{franceschetti2017wave}, this spatial truncation inevitably induces spectral broadening. Specifically, the convolution inherently produces non-zero spectral components that spill beyond the theoretical cutoff wavenumber $\kappa(f)$, causing the observed field to deviate from the ideal band-limited assumption. 
Equivalently, for each frequency offset $f$, this convolution relationship reveals a fundamental insight: reconstructing the latent clean spectrum $\hat{H}_\infty(k_s, f)$ from finite-aperture observations is intrinsically an \textbf{inverse deconvolution problem in the wavenumber domain}.

From a \textit{spatial-domain perspective}, the dense port deployment inherent to FAS induces significant spatial correlation, leading to a pronounced rank deficiency in the channel covariance relative to the total number of sampling points $N$. This deficiency is rigorously governed by the Helmholtz equation \eqref{eq:helmholtz}, which constrains the spatial phase progression of an impinging plane wave to $k_s = \frac{1}{\lambda} \cos \theta$. Consequently, the spatial frequencies of multipath components are strictly confined within the interval $k_s \in [-\frac{1}{\lambda}, \frac{1}{\lambda}]$, defining a two-sided spatial bandwidth of $B_s = \frac{2}{\lambda}$. According to the Landau-Pollak theorem \cite{pizzo2022fourier}, a spatial field band-limited to $B_s$ and truncated to a finite continuous aperture $L$ possesses a finite effective dimensionality. This EDoF establishes a physical upper bound on the statistical rank $r$ of the continuous channel field \cite{11180048}:
\begin{equation}
N_{\text{DoF}} = \lceil B_s L \rceil + 1 = \left\lceil \frac{2 L}{\lambda} \right\rceil + 1 = \lceil 2W \rceil + 1,
\label{eq:edof}
\end{equation}
where $W = L/\lambda$ denotes the normalized aperture size. 
Equation \eqref{eq:edof} provides a closed-form physical upper bound on the intrinsic spatial dimensionality of the continuous-aperture field. The following lemma is an immediate consequence of this EDoF characterization together with a dimension-counting argument for linear recovery.
\begin{lemma}\label{lm1}
Perfect reconstruction of the FAS channel field using linear estimators (e.g., linear minimum mean square error (LMMSE)) requires the number of spatial observations $M$ to satisfy the spatial Nyquist-Shannon sampling criterion dictated by the EDoF:
\begin{equation}
M \ge N_{\text{DoF}} = \lceil 2W \rceil + 1.
\label{eq:sample_complexity}
\end{equation}
\end{lemma}
\begin{proof}
Fix the frequency $f$, and let $\mathcal{V}_f$ denote the $N_{\text{DoF}}$-dimensional effective subspace of admissible spatial fields, where $\dim(\mathcal{V}_f) = \lceil 2W \rceil + 1$ follows from \eqref{eq:edof}. The spatial sampling process is modeled as a linear map $\mathcal{P} : \mathcal{V}_f \to \mathbb{C}^M$. 
Perfect linear reconstruction within this dominant subspace requires $\mathcal{P}$ to be injective, ensuring that distinct fields in $\mathcal{V}_f$ map to unique observation vectors. 
Mathematically, injectivity requires $\operatorname{rank}(\mathcal{P}) = \dim(\mathcal{V}_f) = N_{\text{DoF}}$. Given that the matrix rank is bounded by $\operatorname{rank}(\mathcal{P}) \le \min(M, N_{\text{DoF}})$, the condition $M \ge N_{\text{DoF}}$ is necessary. 
Conversely, if $M < N_{\text{DoF}}$, the rank-nullity theorem dictates that $\dim(\ker(\mathcal{P})) \ge N_{\text{DoF}} - M > 0$. This non-trivial null space induces spatial aliasing, meaning distinct spatial fields yield identical discrete samples, thereby rendering unambiguous reconstruction impossible.
\end{proof}

{\color{black}
This result establishes $N_{\text{DoF}}$ as an aperture-dependent upper bound on the effective spatial dimension of the ideal continuous-aperture field.
When $M<N_{\text{DoF}}$, the observation map has a non-trivial null space on the effective subspace, so measurements alone do not uniquely determine every admissible field.
The proposed operator-learning formulation addresses this underdetermined setting by learning channel-distribution priors together with the non-local mapping suggested by the wave equation; it does not remove the information limitation in Lemma~\ref{lm1}.
Specifically, solving the SPDE in \eqref{eq:spde_system} is equivalent to applying the inverse operator $\mathcal{L}_{\kappa}^{-1}$, which can be mathematically characterized by
\begin{equation}\label{greenfunc}
H(s, f) = \int_{\mathcal{S}} G(s, s'; f) \xi(s', f) d s',
\end{equation}
where $G(s, s'; f)$ represents the Green's-function kernel of the integral transformation.

Crucially, because the spatial aperture $\mathcal{S}$ is strictly bounded, the Green's function is square-integrable over $\mathcal{D}$. 
Accordingly, the integral operator constitutes a compact operator on the Hilbert space. 
Its decaying singular values explain why inversion from noisy finite-aperture observations is ill-conditioned even when the nominal dimension requirement is met.
The duality between the spatial differential operator $\mathcal{L}_{\kappa}$ and its inverse $\mathcal{L}_{\kappa}^{-1}$ therefore motivates a critical shift in perspective: \textit{the FAS reconstruction task can be naturally reframed as an operator learning problem over continuous channel fields rather than a resolution-specific vector recovery problem.}
The resulting neural operator maps sparse measurements to the channel distribution represented in the training data while retaining a non-local integral-operator structure motivated by the Helmholtz model.
}

{\color{black}
\begin{remark}
The projected Helmholtz relation motivates the operator formulation above rather than serving as a pointwise residual constraint during training.
Because $\xi(s,f)$ is an unknown stochastic excitation, imposing the homogeneous residual $\|\mathcal{L}_{\kappa}\widehat H\|_2^2$ would incorrectly force this excitation to zero; moreover, the sparse pilot samples are measurement constraints within $\mathcal{D}$ rather than prescribed Dirichlet or Neumann boundary values on $\partial\mathcal{D}$.
Thus, the proposed UFNO incorporates physics through the continuous-field model and the associated operator structure, without explicit PDE-residual enforcement.
\end{remark}
}

{\color{black}
\begin{remark}\label{rem:edof_interpretation}
The quantity $N_{\text{DoF}}=\lceil 2W\rceil+1$ describes the aperture-dependent dimension of the ideal band-limited spatial field and is invariant to the number of candidate ports used to discretize the aperture.
It upper-bounds the statistical rank $r={\rm rank}(\mathbf{C})$ considered in \cite{11180048}, but the necessary condition $M\geq N_{\text{DoF}}$ does not by itself guarantee stable inversion.
Indeed, for a singular component of the finite-aperture observation operator, $z_i=\sigma_i a_i+n_i$, direct inversion gives $\widehat a_i=a_i+n_i/\sigma_i$; hence, the noise power is amplified by $1/\sigma_i^2$, and an input modal SNR $\rho_i$ becomes $\rho_i\sigma_i^2$ after inversion.
Stable recovery therefore depends jointly on aperture, sampling, operator conditioning, and SNR.
Moreover, the physical singular functions of the finite-aperture operator are not identical to the discrete Fourier coefficients retained by UFNO, so $N_{\text{DoF}}$ is used as a reference for choosing the spatial mode budget rather than as an exact optimal value or a hard performance-saturation point.
\end{remark}
}
\section{Proposed Operator Learning Framework for Continuous SFC Reconstruction}
\label{sec:proposed_method}

In this section, we present the proposed UFNO-based channel estimation framework. We first formulate the channel estimation task as an operator learning problem between infinite-dimensional function spaces to address the continuous aperture nature of FAS. Subsequently, we detail the network architecture design to parameterize this operator efficiently in the spectral domain.

\subsection{Operator Learning Formulation}\label{subsec:problem_formulation}
We propose a resolution-invariant operator framework to map sparse pilot observations to continuous SFC fields, employing the FNO as the underlying architecture. 
To formalize the problem, let $\mathcal{D}$ denote the bounded continuous spatio-frequency domain. We define the input space $\mathcal{A} = \mathcal{A}(\mathcal{D}; \mathbb{R}^{d_a})$ consisting of sparse, noisy observation functions, and the output space $\mathcal{U} = \mathcal{U}(\mathcal{D}; \mathbb{R}^{d_u})$ comprising the corresponding ground-truth continuous channel fields, both are modeled as separable Banach spaces.
Our objective is to optimize a continuous-field operator $\mathcal{G}_{\boldsymbol{\Theta}}: \mathcal{A} \to \mathcal{U}$ that approximates the inverse mapping from the observation space to the SPDE solution manifold, where the optimal parameters $\boldsymbol{\Theta}^{*}$ are determined by minimizing the following expectation
\begin{equation} \label{eq:optimization_objective}
\begin{aligned}
    \arg\min_{\boldsymbol{\Theta}}~ &\mathbb{E}_{H \sim \mu,\, \mathbf{n} \sim \mathcal{N}(0, \sigma^2 \mathbf{I})} \left[ \| \mathcal{G}_{\boldsymbol{\Theta}}(\mathcal{P}(H) + \mathbf{n}) - H \|_{\mathcal{U}}^2 \right] \\
    &\quad + \alpha \mathcal{R}(\boldsymbol{\Theta}),
\end{aligned}
\end{equation}
where $\mu$ denotes the probability measure of the random channel field $H(s,f)$, $\|\cdot\|_{\mathcal{U}}$ denotes the $L^2$-norm on the functional space, and $\mathcal{R}(\boldsymbol{\Theta})$ is a regularization term.

As compared to conventional vector-to-vector regression models, the FNO more naturally captures the continuous nature of the SFC. 
\textcolor{black}{Its spectral convolution provides a non-local parameterization of the reconstruction operator motivated by the Green's-function representation in \eqref{greenfunc}.}
That is, if we modify the spatial sampling grid or increase the port density of the FAS during inference, the neural network should output a consistent continuous channel field $H(s, f)$ without requiring retraining. 
\textcolor{black}{Here, \textit{resolution invariance} means that a learned operator can be evaluated on spatial--frequency discretizations unseen during training, enabling direct cross-resolution channel reconstruction without resolution-specific retraining or explicit output interpolation.}
\textcolor{black}{The finite-aperture spectral structure further motivates spectral-domain reconstruction, while the learnable threshold introduced below suppresses weak, noise-dominated spectral coefficients.}

\begin{figure*}[!t]
\centering
\includegraphics[width=\textwidth]{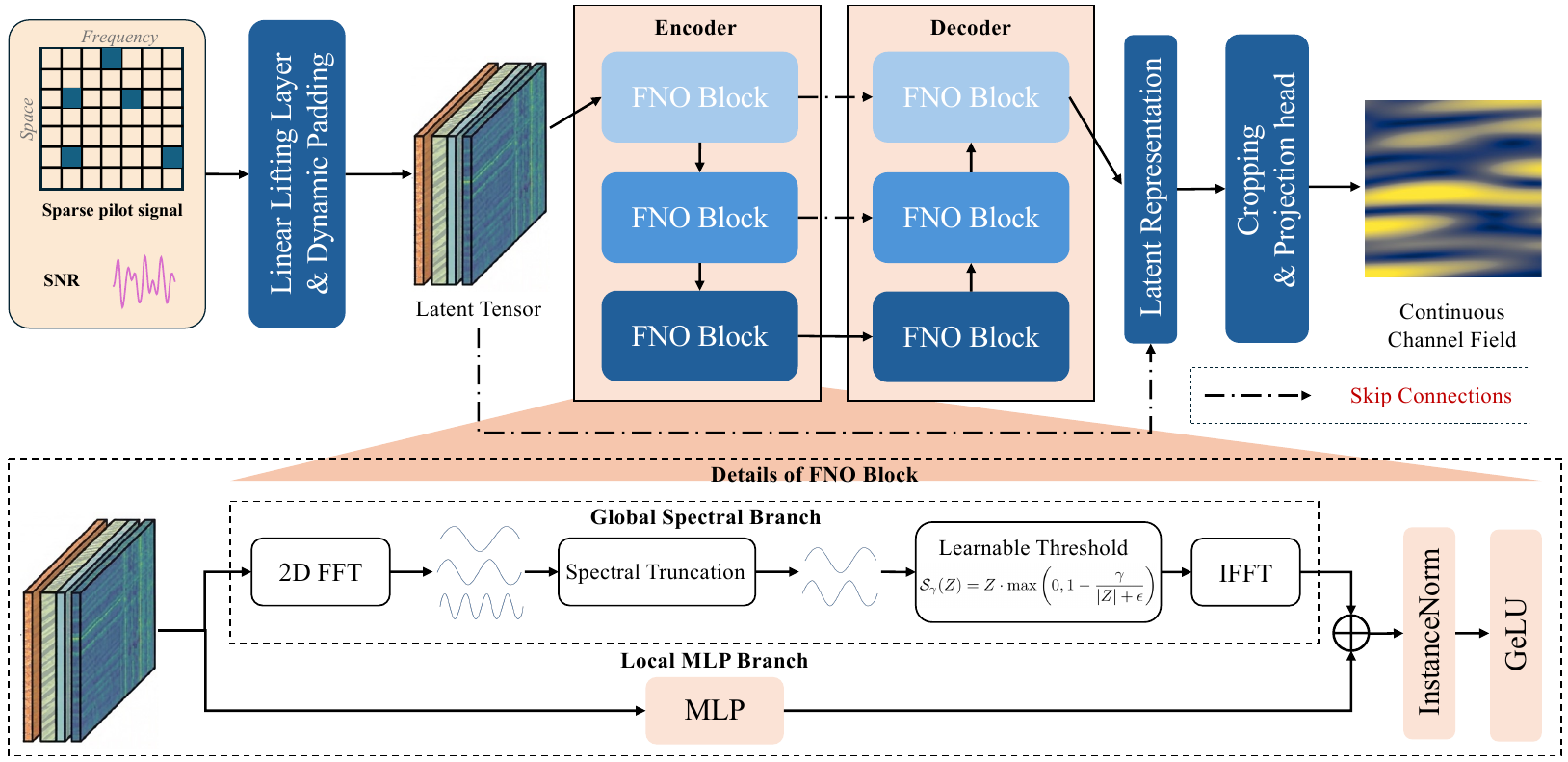}
\caption{Detailed architecture of the proposed UFNO for continuous SFC reconstruction. Top: an end-to-end U-shaped operator pipeline that maps sparse pilot observations to the complex channel field $\hat{H}(s,f)$ through lifting, multi-scale encoding/decoding, and projection; bottom: the internal Fourier block, where a global spectral branch (2D FFT, mode truncation, and learnable SST) is fused with a local MLP branch to jointly capture long-range wave propagation and fine-grained local scattering, while skip connections preserve high-frequency details during reconstruction.}
\label{fig::ufno_net}
\end{figure*}
\subsection{UFNO Architecture Design}
We now detail the end-to-end architecture of the proposed UFNO framework. 
It consists of $2P$ blocks, where the first $P$ blocks form a down-sampling path to extract multi-scale features, and the subsequent $P$ blocks constitute an up-sampling path to reconstruct the high-resolution channel field. 
The architecture is designed to map the sparse pilot observations to the continuous SFC field through three sequential stages: state initialization, dual-branch representation learning, and multi-scale field projection.

\subsubsection{State initialization}:
The data flow begins with the construction of a high-dimensional input tensor from the sparse and noisy observations $\mathbf{y}_{\text{obs}} \in \mathbb{C}^{M}$. 
Specifically, let $N_s$ denote the number of spatial grid points, the input is structured into a multi-channel tensor $\mathbf{X} \in \mathbb{R}^{N_s \times N_f \times 6}$, which encapsulates: 1) the real and imaginary components on the observed pilots; 2) a binary mask identifying the sampled ports; 3) the normalized spatial-frequency grid $(s, f)$ to encode spatial phase references; and 4) a local SNR indicator. 
Subsequently, a point-wise lifting layer $\mathcal{P}: \mathbb{R}^6 \to \mathbb{R}^{d_h}$ is applied to project these raw features into a higher-dimensional latent space of dimension $d_h$, which can be expressed as
\begin{equation}
v_0(s, f) = \operatorname{GELU}(\mathbf{W}_{\text{up}} \mathbf{X}(s, f) + \mathbf{b}_{\text{up}}),
\end{equation}
where $\mathbf{W}_{\text{up}}$ and $\mathbf{b}_{\text{up}}$ denote the learnable weight 
matrix and bias vector, respectively, and $\operatorname{GELU}(x) = x \cdot \Phi(x)$ is the gaussian error linear unit activation \cite{hendrycks2016gaussian}, with $\Phi(\cdot)$ denoting the standard gaussian cumulative distribution function.
{\color{black}Specifically, $\operatorname{GELU}(\cdot)$ provides a smooth transition in the pointwise nonlinear branch rather than the derivative discontinuity of ReLU at the origin, which is generally adopted for derivative-sensitive supervision.}
To mitigate the aperiodic boundary artifacts inherent in the subsequent fast Fourier transform (FFT) operations, the lifted tensor $v_0$ is dynamically zero-padded along both the spatial and frequency dimensions. 
This dynamic padding mechanism provides the framework with size adaptability, enabling it to seamlessly process inputs of varying aperture lengths and subcarrier numbers.
\textcolor{black}{An ablation over $5{,}120$ test channels shows that removing the padding degrades the overall, boundary-region, and interior NMSE by $1.45$, $1.59$, and $1.33$ dB, respectively, confirming that the resulting edge ringing is strongest near the boundaries but also propagates into the interior through the global spectral operator.}

\subsubsection{Dual-Branch Fourier Representation Learning}:
Following the initialization, the padded representation tensor $v_0 \in \mathbb{R}^{S' \times F' \times d_h}$ (where $S'$ and $F'$ denote the dynamically padded spatial and frequency dimensions, respectively) is fed into a cascade of stacked Fourier blocks. 
These blocks serve as the core representation learning engine of the UFNO. 
Each block employs a dual-branch structure to concurrently process the global wave propagation and local scattering details. 
For the $l$-th layer, the representation tensor $v_l \in \mathbb{R}^{S' \times F' \times d_h}$ is updated as
\begin{equation*}
v_{l+1} = \operatorname{InstanceNorm} \left( \sigma  (\mathcal{K}_{\text{global}}(v_l) + \mathcal{W}_{\text{local}}(v_l) )\right),
\end{equation*}
where $\sigma(\cdot)$ is a non-linear activation function, and $\operatorname{InstanceNorm}(\cdot)$ denotes the instance normalization operation applied across the spatial-frequency dimensions. The two parallel branches are meticulously parameterized and physically interpreted as follows:
\begin{itemize}
\item [$\bullet$]\emph{Global Spectral Branch $\mathcal{K}_{\text{global}}(\cdot)$}: This branch implements a parameterized integral operator in the wavenumber domain, designed to extract global macroscopic scattering features. 
To filter high-frequency noise and align with the physical constraints of the propagation environment, the branch first maps the latent representation $v_l$ into the spectral domain through a spectral convolution pipeline (denoted by $\mathcal{V}_l = \mathcal{F}_{\text{mix}}(v_l)$), where we apply a real-to-complex transform on the frequency axis and a Fourier transform on the spatial axis.
{\color{black}
To reduce computational complexity and suppress out-of-band components, we retain a low-mode subset of the mixed spectrum, i.e., $\text{Trunc}(\mathcal{V}_l)$ with spatial and frequency mode budgets.
The implementation parameter $K_s^{\rm side}$ denotes the number of retained spatial Fourier coefficients on each side of the spectrum, and $K_s^{\rm total}=2K_s^{\rm side}$ denotes the corresponding total two-sided budget; $K_f$ denotes the retained frequency-mode budget.
The aperture-dependent $N_{\rm DoF}$ in Lemma~\ref{lm1} provides a physical reference for selecting $K_s^{\rm total}$.
Because the finite-aperture singular functions and the discrete Fourier coefficients are different bases, and because useful recovery also depends on SNR and learned representation capacity, we do not impose $K_s^{\rm total}=N_{\rm DoF}$ as an exact design rule.
}

The retained modes then undergo a learnable complex linear mixing followed by a sparse-promoting denoising operation. To further enhance robustness against in-band receiver noise, we introduce a learnable \textit{spectral soft-thresholding} (SST) mechanism, $\mathcal{S}_{\boldsymbol{\gamma}}(\cdot)$. Here, mode truncation provides explicit low-pass support, while SST performs amplitude-adaptive shrinkage to suppress low-energy coefficients that are typically noise-dominated. The global spectral operation is formulated as:
\begin{equation}
    \mathcal{K}_{\text{global}}(v_l) =  \mathcal{F}^{-1}_{\text{mix}} \left( \mathcal{S}_{\boldsymbol{\gamma}} \left( \mathbf{W} \circledast \text{Trunc}(\mathcal{V}_l) \right) \right),
\end{equation}
where $\circledast$ denotes complex-valued tensor contraction (e.g., Einstein summation) across latent dimensions, and $\mathcal{F}^{-1}_{\text{mix}}(\cdot)$ denotes the inverse mixed transform back to the spatial-frequency grid. The element-wise SST operator $\mathcal{S}_{\boldsymbol{\gamma}}(\cdot)$ for a complex spectral entry $\mathbf{Z}$ is:
\begin{equation}
    \mathcal{S}_{\boldsymbol{\gamma}}(\mathbf{Z}) = \mathbf{Z} \cdot \max\left(0, 1 - \frac{\boldsymbol{\gamma}}{|Z| + \epsilon}\right),
\end{equation}
where $\boldsymbol{\gamma} \in \mathbb{R}_{+}$ is a learnable scalar threshold shared across channels and modes, and $\epsilon$ is a small constant for numerical stability. \textcolor{black}{By integrating this sparsity-inducing threshold into the wavenumber domain, UFNO adaptively suppresses weak, noise-dominated spectral coefficients while retaining the dominant global scattering structure.}

\item [$\bullet$]\emph{Local MLP Branch $\mathcal{W}_{\text{local}}(\cdot)$}: In parallel, a point-wise multi-layer perceptron (MLP), implemented via $1 \times 1$ convolutions, operates strictly along the channel dimension $\mathbb{R}^{d_h} \to \mathbb{R}^{d_h}$ without expanding the receptive field. 
{\color{black}While the spectral branch models non-local spatial--frequency dependencies, spectral truncation can attenuate fine-scale variations. The local branch complements it by representing residual local channel structure and nonlinear feature transformations that are not captured by the retained spectral coefficients.}
\end{itemize}

Compared with standard CNNs or fully connected networks, the UFNO more naturally captures the fundamental physics of electromagnetic wave propagation. 
By explicitly embedding the integral operator structure into the neural architecture, the UFNO is better equipped to model the non-local interactions characteristic of scattering environments. 
In particular, while a standard convolutional layer is limited by a small local receptive field (e.g., $3 \times 3$ kernels), the Fourier layer in the UFNO performs multiplication in the spectral domain, which is mathematically equivalent to a global convolution in the spatial domain. 
This mechanism enables every point in the reconstructed field to instantaneously aggregate information from the entire aperture, effectively mimicking the Huygens-Fresnel principle, where the field at any point is a superposition of wavelets from all sources. 

\begin{remark}
In practical wireless scenarios, sparse observations and additive white Gaussian noise (AWGN) inevitably introduce spurious high-frequency aliasing. The proposed SST operation adaptively filters out low-amplitude spectral components, effectively denoising the SFC field by constraining the solution to reside on the physical wave manifold.
\end{remark}

\subsubsection{U-Shaped Multiscale Fusion}
To construct the complete continuous field, the aforementioned Fourier blocks are organized into a symmetric encoder-decoder U-shaped architecture. 
During the encoding phase, successive SST operations act as sequential low-pass filters, extracting the coarse-grained, low-frequency eigen-modes (e.g., macroscopic angles of arrival and large-scale shadowing). 
During the decoding phase, symmetric skip connections are employed to prevent the attenuation of high-frequency local variations. 
Specifically, the feature map from the $l$-th encoder layer is directly added to the input of the corresponding decoder layer:
\begin{equation}
v_{\text{dec}}^{(l)} = \mathcal{H}_{\text{block}} \left( v_{\text{dec}}^{(l-1)} + v_{\text{enc}}^{(P-l)} \right),
\end{equation}
where $\mathcal{H}_{\text{block}}(\cdot)$ denotes the FNO block in the decoder, and $P$ is the total number of encoder layers. This multi-scale fusion strategy ensures that the precise spatial phase alignment and sub-wavelength small-scale fading details are organically fused with the global scattering structures.
In the final stage, the reconstructed tensor is cropped back to its original physical dimensions to remove the padded boundaries. 
An MLP-based projection head $\mathcal{Q}: \mathbb{R}^{d_h} \to \mathbb{R}^2$ is then applied point-wise to map the high-dimensional latent vectors back to the physical SFC domain.
\subsection{\textcolor{black}{Loss Design and Complexity Analysis}}

\textcolor{black}{The UFNO parameters are learned through supervised offline training and remain frozen during online inference.}

\subsubsection{Sobolev Loss Design}
\textcolor{black}{The UFNO is trained on a large-scale synthetic dataset generated according to the channel model, where the full-field ground-truth SFC, denoted as $H(s, f)$, is available.}
Standard pixel-wise reconstruction objectives, such as the Mean Squared Error (MSE), often yield over-smoothed channel estimates. These methods tend to average out high-frequency variations, which is detrimental for FAS where phase precision is paramount.
To this end, we employ a Sobolev-space loss $\mathcal{L}_{\text{Sobolev}}$ that explicitly constrains both the amplitude and the derivatives of the reconstructed field. The offline training objective is formulated as:
\begin{equation}
\begin{split}
&\mathcal{L}_{\text{Sobolev}} = \| H - \hat{H} \|_{\mathcal{F}}^2 + \beta \mathcal{L}_{\text{Grad}},\\
&\mathcal{L}_{\text{Grad}} =  \| \nabla_s H - \nabla_s \hat{H} \|_{\mathcal{F}}^2 + \| \nabla_f H - \nabla_f \hat{H} \|_{\mathcal{F}}^2,
\end{split}
\end{equation}
where $\|\cdot\|_{\mathcal{F}}$ denotes the Frobenius norm over the discretized spatio-frequency grid, and $\beta$ is a hyperparameter balancing the value and gradient terms. 
\textcolor{black}{In particular, the spatial gradient $\nabla_s H$ characterizes the local spatial phase progression associated with the angles of arrival, while the frequency gradient $\nabla_f H$ characterizes the frequency-selective variation associated with propagation delays. The Sobolev term therefore promotes consistency in physically meaningful spatial--frequency variations; it is a derivative-consistency regularizer rather than a Helmholtz-residual penalty and does not explicitly evaluate $\mathcal{L}_{\kappa}\widehat H-\xi$ during training.}

A distinguishing feature of the proposed UFNO is its \textbf{resolution-invariance}, the learned parameters define a continuous operator independent of the discretization grid. 
However, for practical implementation on digital hardware, the continuous field must be evaluated on a discretized mesh. 
Here, we analyze the computational complexity required to evaluate the operator on a grid of $N = N_s \times N_f$ points, effectively characterizing the model's scalability as the discretization approaches the continuum limit.

\subsubsection{Parameter Complexity}
Crucially, this is independent of the spatial resolution $N$. This property allows the model to be trained on coarse grids and deployed on fine grids without parameter re-scaling.

\subsubsection{Computational Complexity}
We compare the asymptotic floating-point operations (FLOPs) per layer required for a forward pass against state-of-the-art baselines. Here, $N = N_s \times N_f$ denotes the total number of spatio-frequency grid points, and $d_h$ denotes the hidden channel dimension.

{\color{black}\textbf{Proposed UFNO}: The computational bottleneck lies in the 2D FFT/iFFT and the spectral multiplication. For $N$ spatial--frequency samples and $d_h$ hidden channels, the FFTs scale as $\mathcal{O}(d_hN\log N)$. With a fixed retained spectrum of size $K_s^{\rm total}\times K_f$, the spectral multiplication scales as $\mathcal{O}(K_s^{\rm total}K_fd_h^2)$ and is independent of the evaluation resolution. Thus, the per-layer complexity is
\begin{equation}
    \mathcal{C}_{\text{UFNO}} = \mathcal{O}(d_h N \log N + K_s^{\rm total} K_f d_h^2).
\end{equation}
UFNO achieves a \textbf{global receptive field} within a single layer due to the global nature of the Fourier basis.}

\textbf{ViT:} Standard ViTs rely on global self-attention mechanisms to capture long-range dependencies. The complexity encompasses the sequence-to-sequence attention map generation and the linear projections:
\begin{equation}
    \mathcal{C}_{\text{ViT}} = \mathcal{O}(N^2 d_h + N d_h^2).
\end{equation}
While hardware-aware optimizations like FlashAttention \cite{dao2022flashattention} significantly reduce memory IO and latency, the fundamental \textit{algorithmic complexity} remains strictly quadratic w.r.t. the grid size $N$, posing a severe bottleneck for dense spatial sampling.

\textbf{CNN:} Standard convolutional layers with kernel size $K \times K$ scale linearly with $N$ but quadratically with the channel dimension:
\begin{equation}
    \mathcal{C}_{\text{CNN}} = \mathcal{O}(N_{\text{layers}} \cdot N K^2 d_h^2),
\end{equation}
where $N_{\text{layers}}$ is the number of stacked layers. Unlike UFNO and ViT, a single CNN layer only captures local features. To achieve a global receptive field covering the entire $\sqrt{N} \times \sqrt{N}$ grid, the network must stack $N_{\text{layers}} \sim \mathcal{O}(\sqrt{N}/K)$ layers. Consequently, the effective complexity to capture the global multipath correlations scales up to $\mathcal{O}(N^{1.5} K d_h^2)$, losing its linear advantage.

\textbf{Scalability Argument:} In the context of FAS, while the number of hardware-limited RF chains $M$ remains small, characterizing the continuous-aperture potential requires high-resolution reconstruction over a dense discrete grid ($N \gg M$). As we refine the grid size $N$ to approach the continuum limit, the quadratic complexity $\mathcal{O}(N^2)$ of ViT and the depth-dependent penalty $\mathcal{O}(N^{1.5})$ of CNNs become computationally prohibitive for real-time inference. In stark contrast, the \textbf{quasi-linear complexity} $\mathcal{O}(N \log N)$ of the proposed UFNO inherently decouples the global receptive field from the grid resolution, ensuring that achieving ultra-high-resolution channel maps, essential for sub-wavelength port selection, incurs only marginal overhead. This theoretically establishes UFNO as the most viable architecture for real-time continuous-aperture channel estimation.

\section{Simulation Results}\label{SecSim}
\begin{table}[!t]
\centering
\caption{Trainable Parameters of the Proposed UFNO.}
\label{tab:ufno_params}
\renewcommand{\arraystretch}{1.15}
\setlength{\tabcolsep}{3.5pt}
\resizebox{\linewidth}{!}{%
\begin{tabular}{l l r r}
\toprule
\textbf{Component} & \textbf{Dimensionality} & \textbf{\# Params} & \textbf{\%} \\
\midrule
\multicolumn{4}{l}{\emph{Lifting (input $\to$ latent)}} \\
\quad $\mathcal{P}$: Linear & $\mathbb{R}^{6 \times d}$ (+ bias) & 448 & ${<}0.01$ \\
\midrule
\multicolumn{4}{l}{\emph{Encoder / Decoder FNO Blocks ($\times 2L = 10$)}} \\
\quad Spectral weights $\mathbf{R}_1$ & $\mathbb{C}^{d \times d \times k_{s} \times k_{f}}$ & 2{,}097{,}152 & --- \\
\quad Spectral weights $\mathbf{R}_2$ & $\mathbb{C}^{d \times d \times k_{s} \times k_{f}}$ & 2{,}097{,}152 & --- \\
\quad Soft-threshold parameter $\boldsymbol{\gamma}$ & $\mathbb{R}^{1 \times d \times 1 \times 1}$ & 64 & --- \\
\quad \textit{Subtotal: Spectral conv (per block)} & & \textit{4{,}194{,}368} & --- \\[2pt]
\quad MLP skip $\mathbf{W}_1$ (1$\times$1 conv) & $\mathbb{R}^{2d \times d}$ (+ bias) & 8{,}320 & --- \\
\quad MLP skip $\mathbf{W}_2$ (1$\times$1 conv) & $\mathbb{R}^{d \times 2d}$ (+ bias) & 8{,}256 & --- \\
\quad \textit{Subtotal: MLP skip (per block)} & & \textit{16{,}576} & --- \\[2pt]
\quad Instance Norm & (no affine) & 0 & --- \\
\quad \textit{Total per FNO block} & & \textit{4{,}210{,}944} & --- \\
\midrule
All Spectral Conv ($\times 10$ blocks) & & 41{,}943{,}680 & 99.6 \\
All MLP Skip Paths ($\times 10$ blocks) & & 165{,}760 & 0.4 \\
\midrule
\multicolumn{4}{l}{\emph{Projection (latent $\to$ output)}} \\
\quad $\mathcal{Q}_1$: Linear & $\mathbb{R}^{d \times 128}$ (+ bias) & 8{,}320 & 0.02 \\
\quad $\mathcal{Q}_2$: Linear & $\mathbb{R}^{128 \times 2}$ (+ bias) & 258 & ${<}0.01$ \\
\midrule
\textbf{Grand Total} & & \textbf{42{,}118{,}466} & \textbf{100} \\
\bottomrule
\end{tabular}%
}
\end{table}
In this section we validate the theoretical analysis and evaluate the performance of the proposed algorithm. 
\subsection{Simulation Setup}
We evaluate the performance of the proposed UFNO in a geometry-based multipath environment. 
The carrier frequency is set to $3.5$ GHz. The FAS is deployed at the receiver with a physical aperture length of $4\lambda$ (i.e., $W=4$). 
The continuous aperture is discretized into $N_s = 48$ spatial grid points (candidate ports). 
In the frequency domain, we consider an OFDM system with $N_f = 48$ subcarriers and a subcarrier spacing of $15$ kHz. The multipath channel is generated with $N_p = 20$ paths and a delay spread of $300$ ns.
For the pilot arrangement, we employ a uniform sampling strategy with a spacing of $4$ in both spatial and frequency dimensions. During the offline training stage, we generate $51200$ training samples and $5120$ validation samples, where SNR of samples is uniformly sampled from $[-20, 30]$ dB.

\textcolor{black}{The network is trained for $30$ epochs using a batch size of $64$, a learning rate of $5 \times 10^{-3}$, and a weight decay of $1 \times 10^{-5}$. The hyperparameter for $\mathcal{L}_{\text{Grad}}$ is set to $\beta = 0.01$. During the testing phase, the SNR is fixed at $20$ dB unless otherwise specified.}
Besides, we list all trainable components of the proposed UFNO architecture and their dimensionalities in Table~\ref{tab:ufno_params}. 
\textcolor{black}{The encoder--decoder uses $L=5$ FNO blocks per arm and a channel width of $d=64$. The implementation retains $K_s^{\rm side}=16$ spatial Fourier coefficients per side, corresponding to $K_s^{\rm total}=32$ across the two sides of the spatial spectrum, and uses $K_f=16$ frequency modes.}
As shown in Table~\ref{tab:ufno_params}, the spectral convolution path accounts for $99.6\%$ of all trainable parameters, while the point-wise MLP skip connections contribute only $0.4\% (165760$ parameters total across all 10 blocks). Each MLP branch is a two-layer $1{\times}1$ convolution with hidden dimension $2d{=}128$ and GELU activation, serving as a local, high-frequency complement to the global spectral convolution, it is clearly not large enough to perform the full channel estimation task on its own.
The learnable soft-threshold vectors $\boldsymbol{\gamma} \in \mathbb{R}^d$ ($64$ scalars per block, $640$ total) add negligible overhead while enabling adaptive spectral denoising.

To demonstrate the effectiveness of the proposed UFNO framework, we compare its performance against several representative baseline methods, encompassing both traditional signal processing and classical deep learning paradigms:
\begin{itemize}
    \item \textbf{Least Squares (LS):} A non-parametric baseline that applies LS estimation at pilot positions followed by nearest-neighbor interpolation. While computationally efficient, it lacks the ability to exploit global spatio-frequency correlations.
    \item \textbf{OMP / SBL:} Traditional CS and Bayesian recovery methods, respectively. Both utilize a $4\times$ oversampled spatio-frequency dictionary to mitigate off-grid effects. OMP employs greedy residual correlation, while SBL adopts hierarchical Gaussian priors for improved robustness in noisy settings at the cost of higher complexity.
    \item \textbf{DC-GOMP \cite{dong2025frequency}:} A group-greedy extension of OMP that jointly selects the top-$G$ atoms. The group size $G$ is tuned to the spectral leakage width, better capturing spread energy than standard OMP, though it remains sensitive to dictionary-mismatch.
    \item \textcolor{black}{\textbf{CNN:} A four-layer convolutional baseline with channel dimensions $64$, $128$, $128$, and $2$, using $3\times3$ kernels in the hidden layers. It captures local features through nonlinear interpolation.}
    \item \textcolor{black}{\textbf{ViT:} A global-attention baseline that processes the spatial--frequency map as $8\times 8$ patches, with embedding dimension $128$, four attention heads, 10 Transformer layers, and feed-forward dimension $256$.}
    \item \textbf{UFNO w/o SST \& w/o S-Loss:} Ablated variants of the proposed architecture, respectively excluding the learnable SST mechanism and the Sobolev loss.
\end{itemize}
All learning-based methods are trained on the same dataset and evaluated under identical conditions to ensure a fair comparison. \textcolor{black}{The CNN and ViT configurations are fixed representative baselines rather than the outcomes of an exhaustive task-specific hyperparameter search; their dimensions are selected before final test evaluation, and the test set is not used for model selection.} All simulation results are averaged over $5120$ independent channel realizations to ensure statistical significance.

\subsection{\textcolor{black}{Physics-Guided Spatial Mode Selection}}\label{subsec:edof_mode_selection}
{\color{black}
We investigate whether the aperture-dependent EDoF can provide an interpretable physical reference for configuring the spatial spectral capacity of UFNO, rather than selecting the retained Fourier modes solely through empirical tuning.
For $W\in\{2,4,6\}$, corresponding to $N_{\rm DoF}\in\{5,9,13\}$, we hold the $48\times48$ spatial--frequency resolution, $720$-kHz bandwidth, channel model, training-set size, and optimization settings fixed.
A separate UFNO is trained for each tested value of $K_s^{\rm total}$ below, near, or above the corresponding EDoF bound, and each entry in Table~\ref{tab:aperture_modes} is the complex NMSE averaged over $512$ independently generated channels.
At $20$ dB, moving from the tested budget below the EDoF bound to the first tested budget above it improves NMSE by $0.56$, $1.26$, and $1.02$ dB for $W=2,4,6$, respectively; at $30$ dB, the corresponding improvements are $0.69$, $1.64$, and $1.64$ dB.
This consistent behavior across aperture sizes shows that retaining too few spatial modes constrains reconstruction once observation noise is no longer dominant.
At $0$ dB, increasing the mode budget gives no systematic improvement, indicating that recovery is noise limited and cannot be improved simply by retaining more Fourier coefficients.
Additional coefficients beyond the EDoF count can still improve reconstruction because the physical singular functions of the finite-aperture operator are not identical to the discrete Fourier basis used by UFNO.
Therefore, the EDoF translates the physical aperture into an interpretable scale for selecting $K_s^{\rm total}$: it identifies the mode budget required to avoid an undersized spatial representation, while SNR determines how much of that capacity can be usefully exploited.
}

\begin{table}[t]
\centering
\caption{\textcolor{black}{Physics-guided selection of the total two-sided spatial mode budget: complex NMSE across apertures and SNRs (lower is better).}}
\label{tab:aperture_modes}
\setlength{\tabcolsep}{3pt}
{\color{black}
\begin{tabular*}{0.85\linewidth}{@{\extracolsep{\fill}}cccrrr}
\toprule
$W$ & $N_{\rm DoF}$ & $K_s^{\rm total}$ & $0$ dB & $20$ dB & $30$ dB \\
\midrule
2 & 5 & 4 & -15.36 & -27.95 & -28.78 \\
  &   & 6 & -14.96 & -28.50 & -29.47 \\
  &   & 8 & -14.96 & -28.72 & -30.00 \\
\midrule
4 & 9 & 6  & -12.90 & -26.24 & -27.26 \\
  &   & 10 & -12.56 & -27.51 & -28.91 \\
  &   & 14 & -12.42 & -29.12 & -31.96 \\
\midrule
6 & 13 & 10 & -11.59 & -25.31 & -26.50 \\
  &    & 14 & -11.21 & -26.32 & -28.13 \\
  &    & 18 & -11.30 & -26.62 & -28.55 \\
\bottomrule
\end{tabular*}}
\end{table}

\subsection{Performance Comparison}
\begin{figure}[t]
\centering
\includegraphics[width=\linewidth]{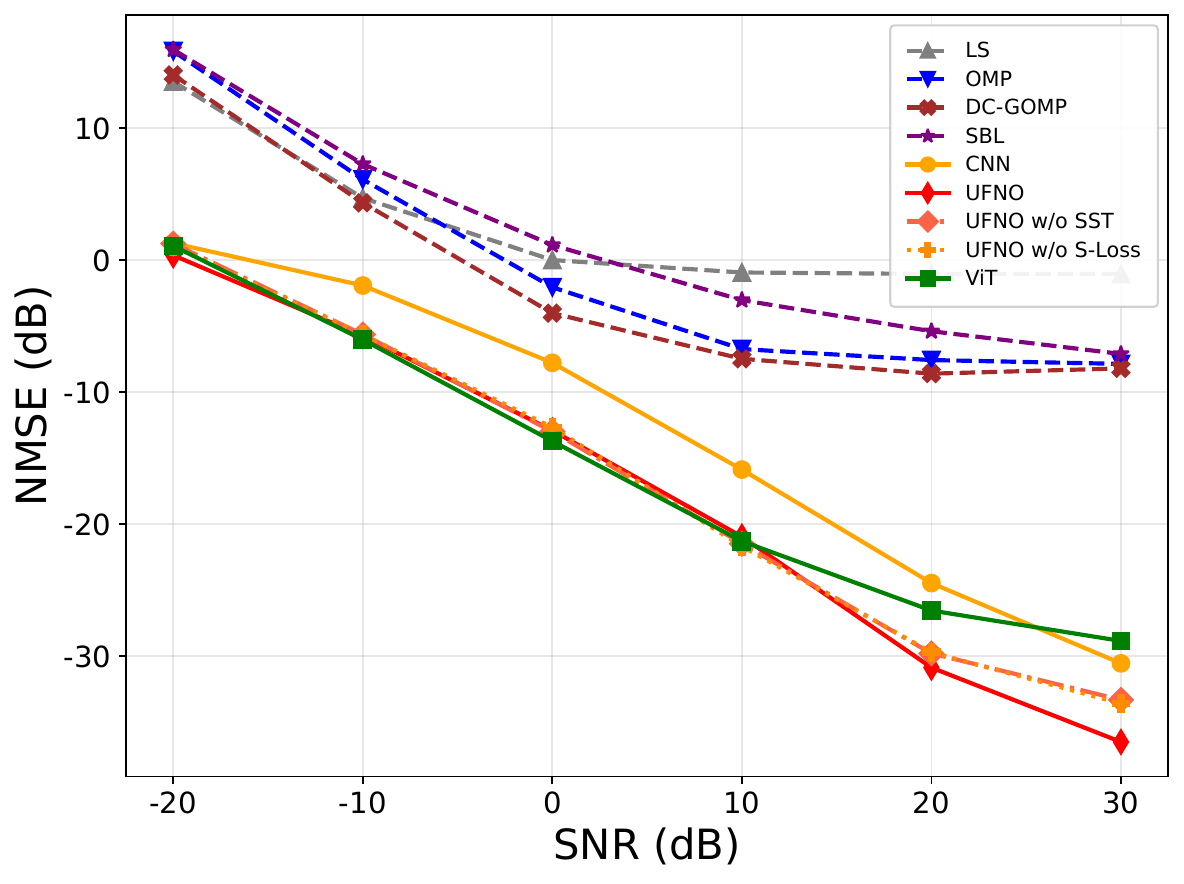}
\caption{NMSE comparison under varying SNR.}
\label{fig::performance_curve_snr}
\end{figure}
\begin{figure}
\centering
\includegraphics[width=\linewidth]{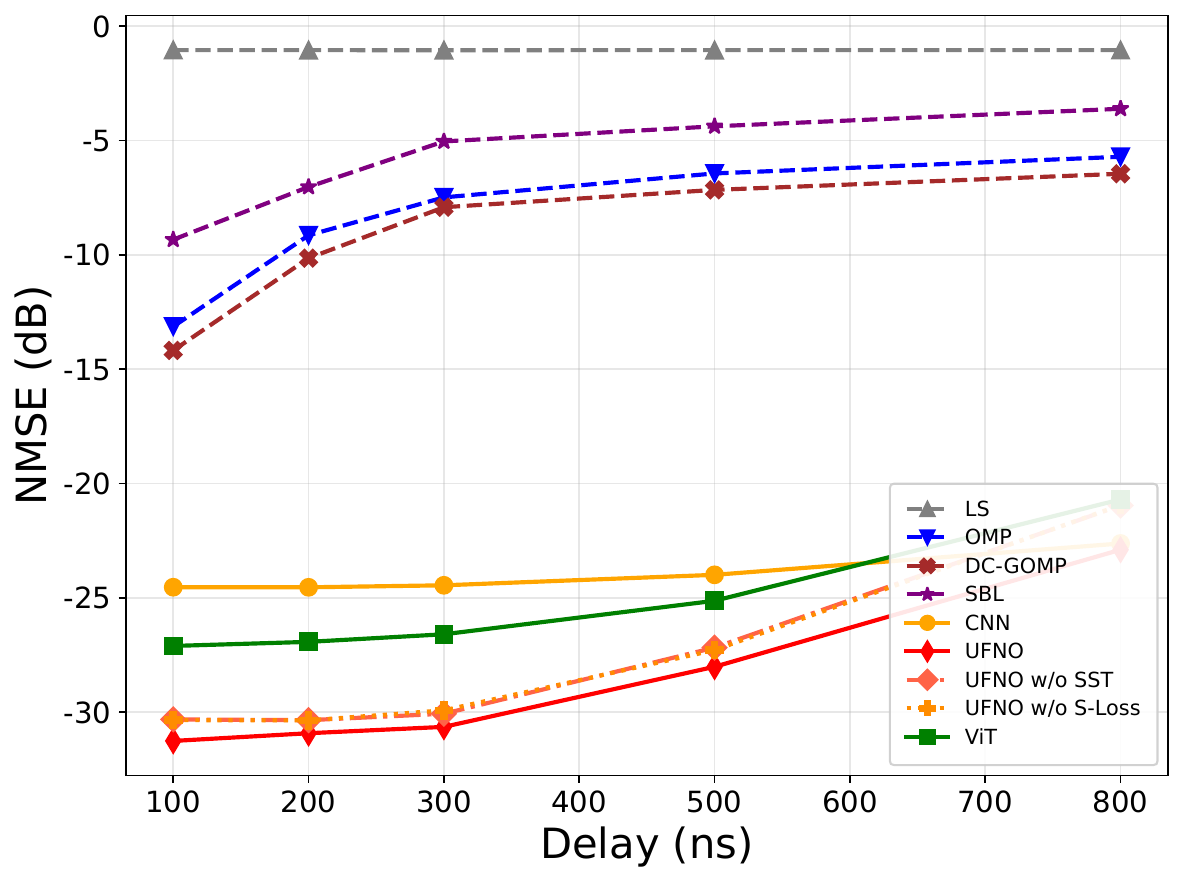}
\caption{NMSE comparison under varying delay spreads (OOD robustness test).}
\label{fig::performance_curve_delay}
\end{figure}
The performance of the proposed UFNO and the benchmark methods is evaluated in terms of the normalized mean squared error (NMSE) between the reconstructed channel $\hat{H}$ and the ground truth $H$, defined as
\begin{equation}
\text{NMSE}(\hat{H}, H) = 10 \log_{10} \left( \frac{\| H - \hat{H} \|_F^2}{\| H \|_F^2} \right).
\end{equation}
Fig. \ref{fig::performance_curve_snr} illustrates the NMSE performance across a wide range of SNR levels, where the proposed UFNO consistently maintains a superior performance envelope. While learning-based architectures significantly surpass traditional model-driven methods by effectively leveraging data-driven field priors, UFNO demonstrates the most robust scaling behavior as the SNR increases. Although UFNO and \textbf{ViT} achieve comparable NMSE in the extremely low-SNR regime, UFNO exhibits a significantly more pronounced performance gain as the noise floor drops, achieving an NMSE below $-35$ dB at $30$ dB SNR and outperforming the \textbf{ViT} and \textbf{CNN} baselines by over $5$ dB. 
The efficacy of our architectural innovations is further elucidated by the performance of the ablated variants in high-SNR environments. As shown in Fig. \ref{fig::performance_curve_snr}, the full UFNO begins to diverge from \textbf{UFNO w/o SST} for SNR $> 10$ dB, reaching a $3$ dB performance margin at $30$ dB. Similarly, \textbf{UFNO w/o S-Loss} trails the full architecture by approximately $3$ dB in the high-SNR regime. These results suggest that while SST is instrumental in adaptively suppressing high-frequency spectral leakage that would otherwise be masked by noise, the Sobolev loss is critical for preserving high-gradient field structures and averting accuracy saturation as the observation quality improves.

Fig.\ref{fig::performance_curve_delay} further evaluates the robustness under increasing delay spread. Physically, a larger delay spread induces faster phase rotations across the frequency domain, intrinsically intensifying the frequency selectivity. 
\textcolor{black}{Under a fixed sampling ratio, this stronger frequency selectivity increases the number of significant frequency-domain components and worsens the conditioning of the underdetermined inverse problem; it does not change the spatial EDoF in \eqref{eq:edof}, which is determined by the aperture.} 
Crucially, as the learning-based models are trained solely on a nominal $300$ ns delay spread, shifting to extremes (e.g., $800$ ns) introduces an out-of-distribution (OOD) shift. Interestingly, \textbf{CNN} exhibits extreme robustness (a nearly flat curve), acting as a localized nonlinear interpolator: its local receptive fields are insensitive to global phase variations, preventing severe OOD collapse albeit at the cost of a poor baseline performance. In contrast, models capturing global dependencies (\textbf{ViT} and UFNO) experience more noticeable degradation under the extreme $800$ ns regime. For UFNO, this behavior is physically interpretable: at enormous delay spreads, the rapid frequency responses shift the intrinsic signal energy into higher spectral modes. 
Consequently, the low-pass mode truncation mechanism (e.g., $K_f = 16$) inherently cuts off these emerging high-frequency channel components. 
We observe that removing either the SST or Sobolev loss modules compromises reconstruction stability as the delay spread deviates from the training set. Notably, \textbf{UFNO w/o S-Loss} exhibits the most pronounced degradation in the high-delay regime (e.g., around $800$ ns), indicating that physics-aware gradient regularization is vital for maintaining field consistency beyond nominal conditions. Nevertheless, the full UFNO maintains a dominant performance advantage over \textbf{ViT} and \textbf{CNN} architectures across the $100$ ns to $500$ ns range. Compared to classical CS-based methods, our framework provides consistent, multi-order-of-magnitude NMSE improvements over the entire evaluated domain, demonstrating the superior generalization power of physics-informed operator learning in complex EM environments.

\subsection{Sampling Ratio}

\begin{figure}[t]
\centering
\includegraphics[width=\linewidth]{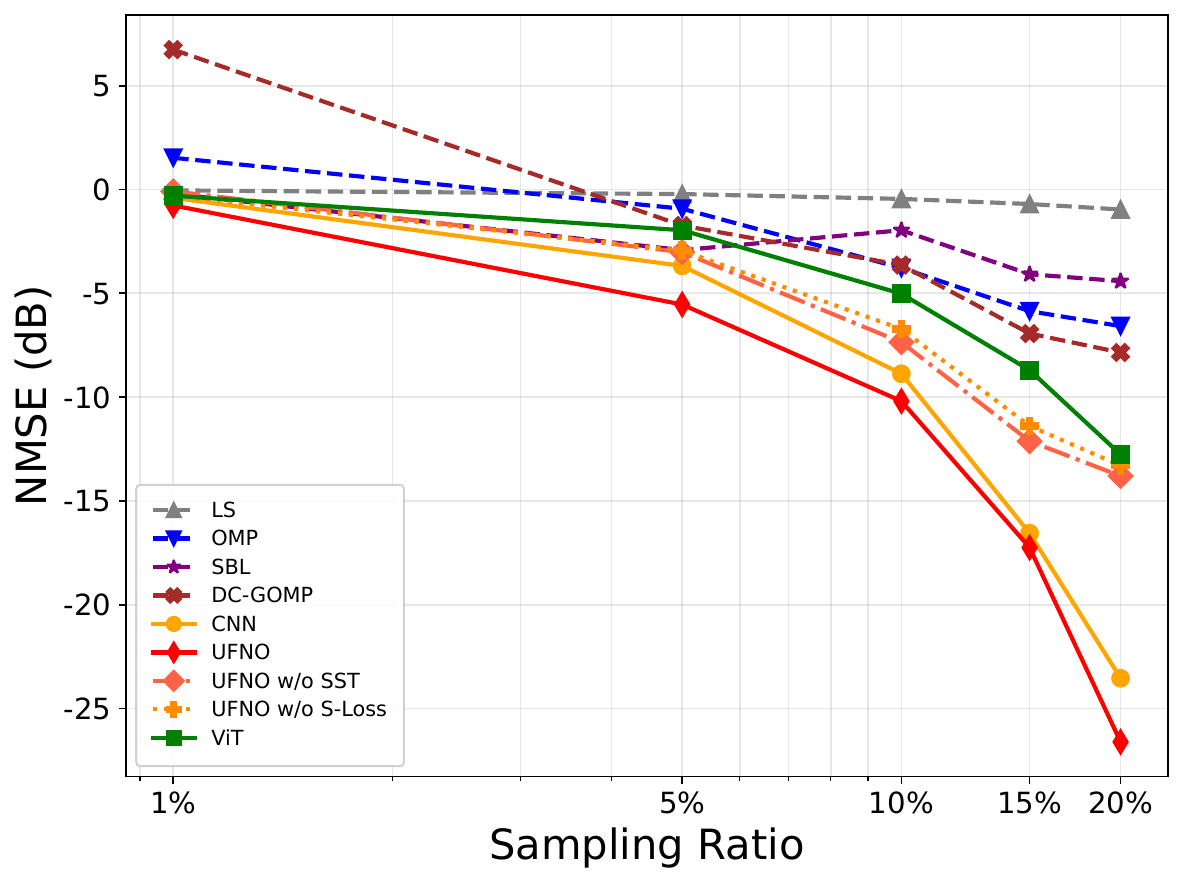}
\caption{NMSE comparison under varying spatio-frequency sampling ratio. }
\label{fig::performance_curve_sampling}
\end{figure}
\begin{figure*}[t]
\centering
\includegraphics[width=\textwidth]{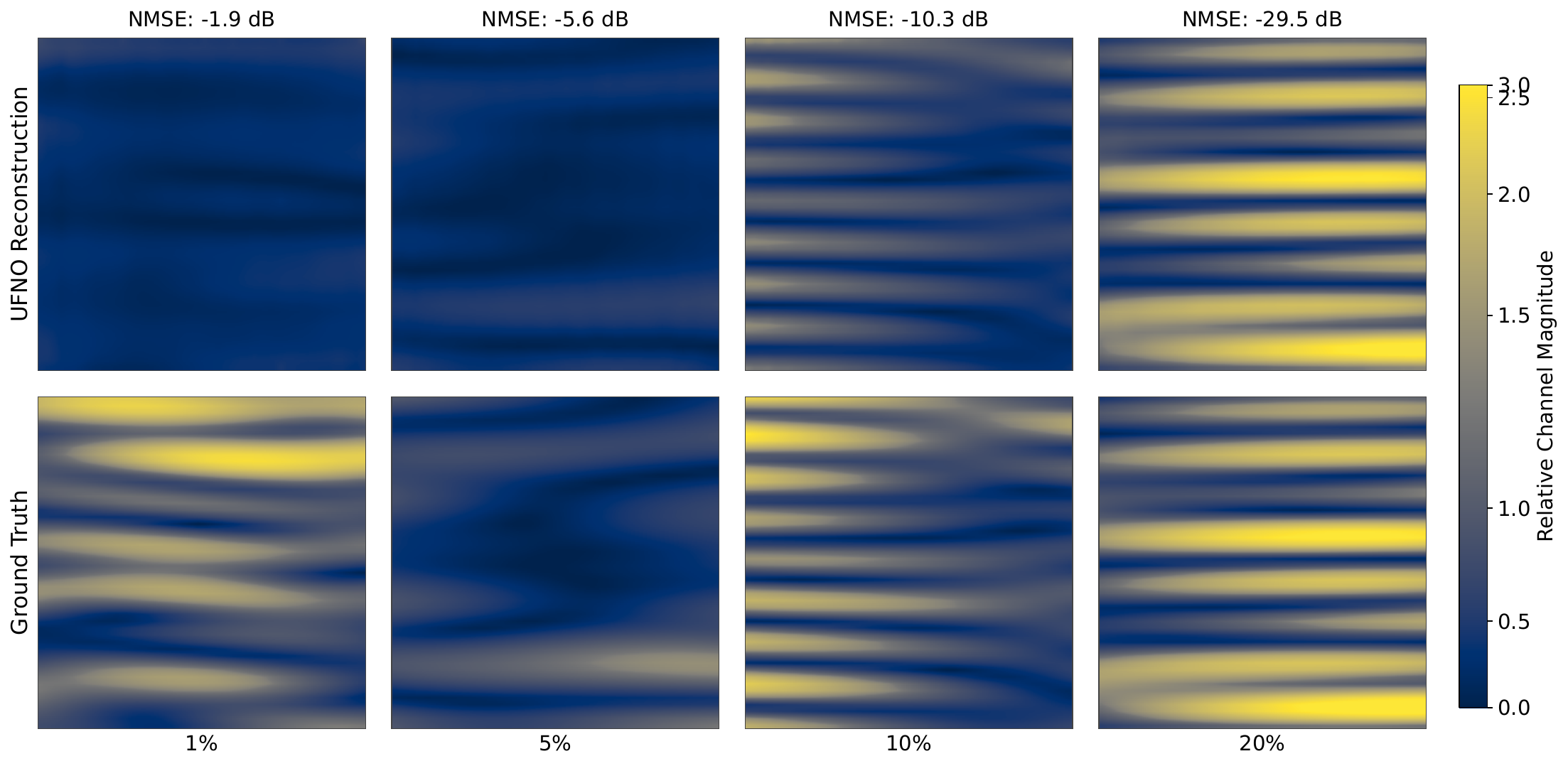}
\caption{Qualitative reconstruction heatmaps under different sampling ratios. \textcolor{black}{All panels show the same randomly selected channel realization, and the annotated NMSE values are realization-specific.}}
\label{fig::sampling_heatmaps}
\end{figure*}

We compare the proposed UFNO with the benchmark methods under varying sampling ratio, defined as 
\begin{equation}\label{eq:sampling_ratio}
\# \text{Sampling Ratio} = \frac{N_{\mathrm{occ}}}{N},
\end{equation} 
where $N_{\mathrm{occ}}$ is the number of occupied spatio-frequency resource blocks, and $N = N_s\times N_f$ represents the number of total available resource blocks in space and frequency domain.
It is worth noting that both uniform and random occupation patterns are evaluated to demonstrate the flexibility of the proposed method in handling arbitrary spatio-frequency resource-block allocations.

As illustrated in Fig.~\ref{fig::performance_curve_sampling}, the NMSE performance of all evaluated methods is presented as a function of the sampling ratio $\in [1\%, 20\%]$. 
\textcolor{black}{Each point on this curve reports the complex NMSE averaged over $5{,}120$ independently generated test channels.}
In extremely sparse regimes, the limited resource-block is fundamentally insufficient to capture the underlying channel manifold, leading to severe ambiguity for all estimators. 
As the sampling ratio increases, traditional model-driven methods struggle to show continuous proportional improvements. 
Specifically, the \textbf{LS} baseline exhibits a nearly flat trajectory due to its naive mapping that fails to spatially interpolate the unobservable zero-padded pcoints. 
Furthermore, compressed sensing methods (\textbf{OMP} and \textbf{SBL}) reveal sub-optimal and marginally jagged performance gains; this is primarily attributed to their rigid sparsity hyperparameter constraints (e.g., fixed maximum number of scatter paths) along with the inherent grid-mismatch uncertainty introduced by discrete uniform masking.

In contrast, deep learning-based methods exhibit significant performance gains as the sampling ratio increases, particularly as sampling ratio approaches the training distribution limit of $21.7\%$. Among these, the proposed UFNO consistently maintains a substantial lead over other learning-based baselines, achieving an NMSE below $-26$ dB at $\#\text{sampling ratio} = 20\%$, which significantly outperforms \textbf{CNN} and \textbf{ViT} models. 
This superiority is further elucidated by the performance of the ablated variants: as sampling ratio increases, the gap between the full UFNO and its counterparts widens markedly. Specifically, at $20\%$, the full model outperforms \textbf{UFNO w/o SST} and \textbf{UFNO w/o S-Loss} by approximately $13$--$14$ dB. This substantial margin underscores the necessity of adaptive spectral truncation and Sobolev-gradient supervision in highly under-determined reconstruction, confirming that the synergistic interplay between spectral manifold partitioning and gradient-level fidelity is critical for capturing fine-grained sub-wavelength structures.

Fig.~\ref{fig::sampling_heatmaps} provides qualitative evidence of how reconstruction fidelity evolves with sampling density. As sampling ratio increases from $1\%$ to $20\%$, the recovered field transitions from coarse low-frequency envelopes to structured multipath textures with progressively improved stripe orientation and amplitude contrast, with NMSE improving from approximately $-1.9$ dB to $-29.5$ dB, indicating that dominant spatio-frequency components become reliably recoverable beyond very sparse sampling.

\subsection{Resolution Invariance}

\begin{table*}[t]
\centering
\caption{\textcolor{black}{Controlled cross-resolution complex NMSE (dB). All learned estimators are trained only at $48\times48$ and frozen during evaluation. (lower is better.)}}
\label{tab:cross_resolution}
\color{black}
\renewcommand{\arraystretch}{1.12}
\scriptsize
\setlength{\tabcolsep}{3.5pt}
\begin{tabular*}{\textwidth}{@{\extracolsep{\fill}}c*{9}{r}@{}}
\toprule
& \multicolumn{4}{c}{Learned estimators} & \multicolumn{3}{c}{Sparse-pilot interpolation} & \multicolumn{2}{c}{$48\!\times\!48$ CNN upsampling} \\
\cmidrule(lr){2-5}\cmidrule(lr){6-8}\cmidrule(lr){9-10}
Spatial $\times$ frequency & \textbf{UFNO} & Standard FNO & CNN & ViT & Nearest & Linear & Cubic & Bilinear & Bicubic \\
\midrule
$48\times48$   & \textbf{-30.67} & -29.83 & -24.47 & -26.80 & -10.76 & -17.89 & -20.12 & -24.47 & -24.47 \\
$60\times60$   & \textbf{-29.55} &  -8.28 & -24.47 &  -9.45 & -12.50 & -20.39 & -21.04 & -24.47 & -24.44 \\
$72\times72$   & \textbf{-28.87} &  -7.09 & -24.13 &  -5.90 & -13.68 & -21.21 & -21.06 & -24.48 & -24.44 \\
$96\times96$   & \textbf{-29.50} &  -6.48 & -24.40 & -10.16 & -15.48 & -21.88 & -21.07 & -24.50 & -24.46 \\
$120\times120$ & \textbf{-29.46} &  -8.18 & -24.51 &  -5.82 & -16.64 & -22.06 & -21.08 & -24.50 & -24.47 \\
$144\times144$ & \textbf{-29.56} &  -8.55 & -24.38 &  -8.94 & -17.40 & -22.13 & -21.08 & -24.50 & -24.47 \\
\bottomrule
\end{tabular*}
\end{table*}
{\color{black}
We evaluate resolution invariance through a controlled fixed-domain comparison and a broader cross-resolution stress test. The controlled comparison isolates changes in the spatial--frequency discretization, whereas the stress test additionally examines increasingly large departures from the training configuration.
To isolate discretization transfer from changes in the propagation setting, we conduct the controlled cross-resolution comparison in Table~\ref{tab:cross_resolution}. For each of 512 independently generated channels, the same underlying continuous SFC is evaluated at all spatial--frequency resolutions. All learned estimators are trained only at $48\times48$ and applied without parameter updates or resolution-specific fine-tuning. The controls comprise direct inference by CNN, ViT, and Standard FNO; nearest, linear, and cubic interpolation from the target-resolution sparse pilots; and bilinear or bicubic upsampling of the dense $48\times48$ CNN reconstruction.
For the sparse-pilot interpolation baselines, nearest interpolation assigns each unobserved location the value of its closest pilot, whereas linear and cubic interpolation construct piecewise linear and smooth piecewise cubic surfaces, respectively, directly from the observed pilot coordinates. For the upsampling baselines, the CNN first produces a dense $48\times48$ channel estimate; its real and imaginary components are then resized to the target resolution using bilinear interpolation over adjacent output samples or bicubic interpolation over a larger local neighborhood. The former group therefore reconstructs directly from the target-resolution pilots, whereas the latter tests post-processing of a low-resolution learned channel estimate.

Across the five unseen resolutions, UFNO achieves an average NMSE of $-29.39$~dB, with only a $1.28$~dB average loss relative to its native-resolution result. 
It outperforms direct CNN, linear interpolation, cubic interpolation, and direct Standard FNO by $5.01$, $7.86$, $8.32$, and $21.67$~dB on average, respectively, and exceeds CNN bilinear and bicubic upsampling by $4.90$ and $4.93$~dB. Standard FNO attains $-29.83$~dB at $48\times48$ but degrades sharply at every unseen resolution, so its loss cannot be attributed to an inaccurate native-resolution model. Although both models employ Fourier weights that are independent of the spatial--frequency sampling size, this property permits evaluation on a new discretization but does not guarantee unchanged finite-resolution accuracy. Standard FNO propagates changes in the discrete spectrum of the sparse-pilot input through ten sequential Fourier blocks, whereas UFNO's symmetric inter-layer skip connections provide short paths for intermediate representations and its SST suppresses low-energy spectral coefficients that are more sensitive to sampling changes and noise. These mechanisms limit the accumulation of discretization mismatch and provide a structural explanation for UFNO's more stable cross-resolution performance. These controlled comparisons demonstrate the intended benefit of the proposed operator-learning framework: a single frozen UFNO preserves high-fidelity SFC reconstruction across unseen discretizations through direct inference, rather than through resolution-specific retraining or post-hoc resizing.
}

\begin{figure*}[!t]
\centering
\includegraphics[width=\textwidth]{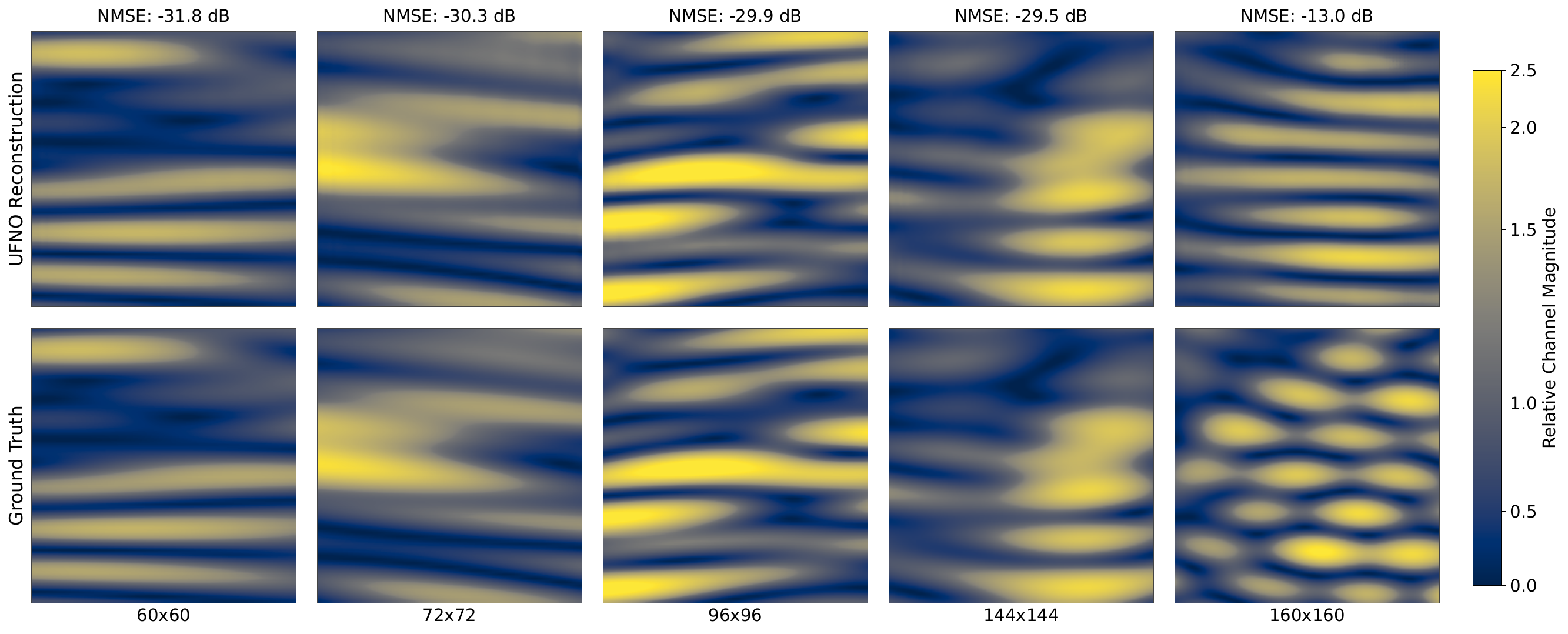}
\caption{\textcolor{black}{Qualitative cross-resolution reconstruction with a single frozen UFNO. Each column shows one representative channel realization at the indicated spatial--frequency resolution, and the annotated NMSE is realization-specific. The reconstructed fields remain smooth and preserve the dominant channel structures from $60\times60$ to $144\times144$, whereas visible mismatch emerges at $160\times160$.}}
\label{fig:multires_heatmaps}
\end{figure*}

{\color{black}
Fig.~\ref{fig:multires_heatmaps} visualizes representative one-shot UFNO reconstructions at unseen spatial--frequency resolutions. From $60\times60$ to $144\times144$, the reconstructed fields closely follow the ground truth, exhibiting smooth spatial--frequency variation and continuous multipath structures while retaining their dominant fine-scale features. This qualitative behavior is consistent with the proposed design: the global Fourier branch and learnable SST maintain coherent large-scale channel structure and suppress noise-dominated spectral components, the local branch and U-shaped skip connections preserve fine variations, and the Sobolev loss promotes derivative consistency across space and frequency.

At $160\times160$, visible high-frequency mismatch marks the limit of this transfer range. The retained Fourier-mode budget then covers a smaller relative portion of the spectrum, while the enlarged physical bandwidth under fixed subcarrier spacing further shifts the test setting away from the $48\times48$ training distribution. Thus, the visualization supplements Table~\ref{tab:cross_resolution} by showing both the smooth cross-resolution reconstruction enabled by UFNO and the onset of mismatch under a substantially larger distribution shift.
}

\subsection{\textcolor{black}{Computational Efficiency}}
{\color{black}
We evaluate the online inference overhead.
The OMP and SBL implementations execute on CPU and do not invoke GPU acceleration, whereas the four learned estimators execute on an NVIDIA GeForce RTX 4090 D GPU.
}

\begin{table*}[t]
\centering
\caption{\textcolor{black}{Online inference latency and offline training cost.}}
\label{tab:computational_overhead}
{\color{black}
\setlength{\tabcolsep}{6pt}
\textbf{(a) Online inference latency (ms)}\\[1mm]
\begin{tabular*}{0.85\linewidth}{@{\extracolsep{\fill}}llrrrrrr}
\toprule
& & \multicolumn{6}{c}{Spatial $\times$ frequency resolution} \\
\cmidrule(lr){3-8}
Method & Device & $48\!\times\!48$ & $60\!\times\!60$ & $72\!\times\!72$ & $96\!\times\!96$ & $120\!\times\!120$ & $144\!\times\!144$ \\
\midrule
UFNO & GPU & 9.57 & 9.68 & 9.68 & 9.53 & 9.58 & 9.51 \\
Standard FNO & GPU & 8.29 & 8.43 & 8.43 & 8.21 & 8.32 & 8.28 \\
CNN & GPU & 0.53 & 0.53 & 0.53 & 0.61 & 0.62 & 0.68 \\
ViT & GPU & 2.05 & 2.13 & 2.09 & 2.03 & 2.05 & 2.07 \\
\midrule
OMP & CPU & 21.64 & 22.18 & 16.35 & 47.07 & 98.95 & --- \\
SBL & CPU & 108.11 & 4542.3 & 8035.8 & 38153 & 182090 & --- \\
\bottomrule
\end{tabular*}

\vspace{2mm}
\textbf{(b) Offline training cost of the learned estimators}\\[1mm]
\setlength{\tabcolsep}{5pt}
\begin{tabular*}{0.78\linewidth}{@{\extracolsep{\fill}}lrrrr}
\toprule
Method & Best epoch & Time to best & s/epoch & Peak VRAM \\
\midrule
UFNO & 28 & $39.68$ min & $85.04$ & $5.21$ GB \\
Standard FNO & 30 & $39.33$ min & $78.67$ & $4.61$ GB \\
CNN & 25 & $5.90$ min & $14.16$ & $0.63$ GB \\
ViT & 25 & $7.18$ min & $17.27$ & $0.22$ GB \\
\bottomrule
\end{tabular*}}
\end{table*}

{\color{black}
Table~\ref{tab:computational_overhead} consolidates the online latency of all six methods and identifies the execution device of each implementation.
The four learned estimators use the same GPU-side protocol with batch size one, float32 precision, explicit CUDA synchronization, 50 warm-up iterations, and 200 timed iterations. 
UFNO maintains an almost resolution-independent latency of $9.5$--$9.7$ ms, whereas OMP and SBL increase from $21.6$ and $108.1$ ms at $48\times48$ to $99.0$ ms and $182.1$ s at $120\times120$, respectively.
Thus, UFNO is approximately $2.3\times$ and $11.3\times$ faster than OMP and SBL, respectively, at $48\times48$, and these measured advantages increase to approximately $10.3\times$ and $1.9\times10^{4}\times$ at $120\times120$.
Panel (b) reports one complete 30-epoch profiling run for each learned estimator under the same settings.
Among the learned estimators, the current UFNO implementation has higher latency than Standard FNO, ViT, and CNN, but provides substantially better reconstruction accuracy and cross-resolution generalization than CNN and ViT.
In the current PyTorch implementation, the higher latency mainly arises from the repeated IFFT operations in the Fourier blocks.
In contrast, the convolution operations in CNN and dense matrix operations in ViT benefit from highly optimized GPU implementations.
}

\section{Conclusion}
This paper establishes a paradigm shift for broadband FAS channel estimation by reformulating the ill-posed discrete recovery task as a continuous physical field reconstruction problem. \textcolor{black}{The formulation is motivated by the Green's-function integral structure of the projected Helmholtz model and reveals that spectral reconstruction is fundamentally an inverse deconvolution problem.} To address this, we propose the UFNO framework, which incorporates a learnable SST mechanism and a dual-branch architecture to circumvent inherent low-pass spectral biases. \textcolor{black}{Furthermore, a physics-aware Sobolev loss promotes consistency in both channel values and physically meaningful spatial--frequency variations.} Extensive numerical results demonstrate that the proposed operator learning framework substantially outperforms both traditional iterative solvers and state-of-the-art deep learning baselines, delivering superior reconstruction fidelity alongside \textcolor{black}{orders-of-magnitude lower inference latency than iterative solvers at high spatial--frequency resolutions}.

\textcolor{black}{Future work will validate UFNO using measured channels and investigate its integration into practical receivers.
Its cross-resolution inference capability may further facilitate deployment with different spatial and frequency sampling resolutions.
The framework may also be extended to planar or volumetric FAS configurations. Supporting time-varying channels requires a spatio-temporal operator that captures Doppler effects and temporal correlation.
The fixed Fourier-basis representation adopted in the present spatial--frequency formulation may be insufficient to characterize the coupled non-stationarity across these dimensions, and more general tools from harmonic analysis may be useful for developing suitable representations and operators.}

\bibliographystyle{IEEEtran}
\bibliography{references}
\end{document}